\documentclass[11pt]{article}

\usepackage[T1]{fontenc}
\usepackage[margin=1in]{geometry}
\usepackage{amsmath,amssymb,amsthm,mathtools}
\usepackage{booktabs}
\usepackage{float}
\usepackage[vlined,linesnumbered,ruled]{algorithm2e}
\usepackage{microtype}
\usepackage{tabularx}
\usepackage[table]{xcolor}
\usepackage{soul}
\definecolor{addyellow}{rgb}{1,1,0.70}
\definecolor{delgreen}{rgb}{0.80,0.95,0.80}

\newsavebox{\ClaudeBox}

\soulregister{\cite}{7}
\soulregister{\ref}{7}
\soulregister{\eqref}{7}
\usepackage[hidelinks]{hyperref}

\newcolumntype{L}[1]{>{\raggedright\arraybackslash}p{#1}}
\newcolumntype{Y}{>{\raggedright\arraybackslash}X}

\newtheorem{theorem}{Theorem}[section]
\newtheorem{lemma}[theorem]{Lemma}
\newtheorem{fact}[theorem]{Fact}
\newtheorem{corollary}[theorem]{Corollary}
\theoremstyle{remark}
\newtheorem{remark}[theorem]{Remark}
\newtheorem*{thmA}{Theorem A (informal)}
\newtheorem*{thmB}{Theorem B (informal)}

\def\inline#1:{\par\vskip 7pt\noindent{\bf #1:}\hskip 10pt}
\def\SamplerBalance{\mbox{\sf Sampler\_Balance}}
\def\SlackGreedy{\mbox{\sf Slack\_Greedy}}
\def\CONGEST{\mbox{\tt CONGEST}}
\def\LOCAL{\mbox{\tt LOCAL}}
\def\CC{\mbox{\tt CC}}
\def\LOCAL{\mbox{\tt LOCAL}}
\def\cE{\mathcal{E}}
\def\cF{\mathcal{F}}
\def\cP{\mathcal{P}}
\def\clr{\phi}
\def\freq{{\tt f}}
\def\numclr{\chi}
\def\TVD{\mbox{\tt TVD}}
\def\Prob{{\mathbb P}}
\def\Exp{{\mathbb E}}
\soulregister{\CONGEST}{0}
\soulregister{\CC}{0}
\soulregister{\LOCAL}{0}
\soulregister{\SamplerBalance}{0}
\soulregister{\SlackGreedy}{0}
\soulregister{\TVD}{0}

\title{
Diameter-Free Distributed Frequency Control for Graph Coloring
\\
in the \CONGEST\ Model}

\author{Amit Nir\thanks{Weizmann Institute of Science.
E-mail: \texttt{\{amit.nir,david.peleg\}@weizmann.ac.il}.}
\and
David Peleg$^*$}

\date{\today}

\begin{document}
\maketitle

\begin{abstract}
This paper presents two randomized proper-coloring algorithms that control
color frequencies in the synchronous \CONGEST\ model without paying a
diameter-dependent coordination cost.  Let $\lambda\geq 1$ denote the desired
failure exponent.  For every fixed $\delta>0$, the first algorithm uses
$\numclr=\lceil(2+\delta)\Delta\rceil$ colors and, with probability at least
$1-n^{-\lambda}$, outputs a proper coloring that bounds the deviation of every color frequency from $n/\numclr$ by
$O_\delta\!\left(
 \sqrt{(\lambda+1)(n/\numclr)\lg n}+(\lambda+1)\lg n
\right).$
Under an explicit load condition, this additive guarantee yields two-sided
relative balance.  The second algorithm works with every
$\numclr>\Delta$ and gives a one-sided frequency cap controlled by the palette
slack $\numclr-\Delta$.  In particular, it uses
$\Delta+\lceil(\Delta+1)/\lceil\ln n\rceil\rceil$ colors and caps every used
color class by
$O((\lambda+1)(\sigma\lg^2 n+\lg n))$, where
$\sigma=n/(\Delta+1)$.
Both algorithms run in $O((\lambda+1)\lg n)$ rounds, with no dependence
on the network diameter; for the first algorithm, the multiplicative constant in the time bound depends on $\delta$.  
%whereas the second algorithm involves no parameter $\delta$ at all.
\end{abstract}

%%%%%%%%%%%%
\section{Introduction}
\label{sec:intro}

%%%%%%%%%%%%
\subsection{Background and Problem}

We study distributed algorithms for \emph{equitable graph coloring}. The problem is governed by two different requirements: \emph{properness} requires adjacent vertices to use different colors, which can be enforced by \emph{local} communication, while  \emph{balance} (or \emph{equitability}), requires that the \emph{frequency} of each color (i.e., the number of vertices assigned to it) be (almost) equal, which is inherently \emph{global}. The more relaxed \emph{near-equitable} coloring problem asks whether these conflicting requirements can be met simultaneously, and at what cost in palette size, frequency discrepancy, and running times.

To illustrate why exact balance is a global constraint, consider an anonymous
ring $C_n$ and any randomized $t$-round algorithm on it.  By symmetry, all
vertices share one output distribution, which properness forces to be
non-degenerate, and the outputs of any $\lfloor n/(2t+1)\rfloor$ vertices at
pairwise distance exceeding $2t$ are mutually independent, being functions of
disjoint portions of the randomness.  The number of such vertices receiving a
fixed color therefore fluctuates by order $\sqrt{n/t}$.  This suggests that
maintaining all frequencies within $\pm1$ of each other, as exact equitability
demands, requires coordination at diameter scale, whereas properness alone is
a purely local matter.  Accordingly, we relax exactness and quantify how
closely balance can be approximated without global coordination.  (For formal lower bounds see \cite{NirPeleg2026NearEquiLocal}.)

Let $G=(V,E)$ be a finite simple undirected graph.  Denote
$n=|V|$, $m=|E|$, and let $\Delta$ be the maximum degree.  We use $D$ to
denote the largest diameter of a connected component of $G$; thus $D$ remains
finite even when $G$ is disconnected.  The open neighborhood of a vertex
$v$ is denoted by $N(v)$, and $\deg(v)=|N(v)|$.

For an integer $\numclr\geq 1$, a \emph{$\numclr$-coloring} is a function
$\clr:V\to[\numclr]$, where $[\numclr]=\{1,\ldots,\numclr\}$ is termed the
\emph{palette}.  The coloring is \emph{proper} if
$\clr(u)\ne\clr(v)$ for every $uv\in E$.  For a color
$i\in[\numclr]$, its \emph{frequency} is
$\freq(i)=|\{v\in V\mid \clr(v)=i\}|$.
A color of positive frequency is termed a \emph{used color}.  Denote the
ideal frequency associated with the palette by
$\sigma_{\numclr}=n/\numclr$, and put $\sigma=n/(\Delta+1)$.

We distinguish three balance notions.  A proper coloring is
\emph{equitable} if the frequencies of any two palette colors differ by at
most one; equivalently, every frequency belongs to
$\{\lfloor n/\numclr\rfloor,\lceil n/\numclr\rceil\}$.  For $B\geq 0$, it is
\emph{$B$-additively balanced} if
$|\freq(i)-\sigma_{\numclr}|\leq B$ for every palette color $i$.  For
$0<\eta\leq 1$, it is \emph{$\eta$-relatively balanced} if
\[
(1-\eta)\sigma_{\numclr}\leq \freq(i)\leq
(1+\eta)\sigma_{\numclr}
\qquad\text{for every }i\in[\numclr].
\]
These two relaxations are referred to as \emph{near-equitable}.  Finally, a
proper coloring is \emph{$F$-upper-balanced} if every used color has frequency
at most $F$.  An upper-balanced coloring need not be near-equitable, since it
may contain singleton color classes.

The classical Hajnal--Szemer\'edi theorem ensures that an equitable
$\numclr$-coloring exists for every
$\numclr\geq\Delta+1$~\cite{HajnalSzemeredi1970}; polynomial-time centralized
constructions are also known~\cite{KiersteadEtAl2010}.  The current paper
addresses a different issue, namely, obtaining explicit frequency guarantees
quickly in bandwidth-limited distributed networks.

%%%%%%%%%%%%
\subsection{Distributed Models and Conventions}

We work in the synchronous \CONGEST\ model.  The communication network is
$G$, each vertex hosts a processor driven by a global clock, and in every
round each endpoint of an edge may send an $O(\lg n)$-bit message to the
other endpoint.  
%We also record the corresponding bounds for the \emph{Congested Clique} (\CC) model.  In that model the processors are the vertices of $G$, the communication network is a clique, and every ordered pair of processors may exchange an $O(\lg n)$-bit message in every round. The input graph itself is still $G$.

The algorithms assume common knowledge of $n$, $\Delta$, the palette size $\numclr$,
and their accuracy and failure parameters.  The \emph{failure exponent}
$\lambda\geq1$ is an input parameter; the phrase \emph{with high probability}
(w.h.p.) means probability at least $1-n^{-\lambda}$ for the specified value
of $\lambda$.  We use $\lg$ for the base-$2$ logarithm and $\ln$ for the
natural logarithm.  All probability statements concern the internal random
choices of the algorithms.

The exact value of $n$ can be replaced by a common estimate $N$ satisfying
$n\leq N\leq n^{K}$ for a fixed constant $K$.  Specifically, replacing $n$
by $N$ in stopping times and failure parameters changes the asymptotic round
bounds by at most a constant factor and can only decrease the stated failure
probabilities.
The same replacement enters the frequency parameters as well: the
quantities $R_{\mathrm{samp}}$, $B_{\numclr}$, $T$, $R_{\mathrm{prop}}$ and
$\Psi_{p,\numclr}$ defined below are then computed with $N$ in place of $n$,
and since $N\leq n^{K}$, each of them grows by at most the constant factor
$K$; hence all stated deviation bounds and frequency caps degrade by at most
a constant factor as well.
To avoid cumbersome notation, the presentation assumes that
$n$ itself is known.

When $\Delta=0$, assigning the single color $1$ to every vertex gives an
equitable coloring in zero rounds.  Hence we assume $\Delta\geq1$ hereafter.
No connectedness assumption is made.

%%%%%%%%%%%%
\subsection{Our Contributions}

Our earlier work~\cite{NirPeleg2026NearEquitable} obtains several palette--frequency tradeoffs by maintaining global color counts and assigning global color quotas.
In the \CONGEST\ model, it achieved (1) a coloring with at most $2(\Delta+1)$ colors, frequency range $[\sigma/2,\sigma]$ and time $O(\lg\Delta\,(D+\Delta+\lg^5\lg n))$, (2) for every integer $k\ge1$, a coloring with at most $(1+1/k)(\Delta+1)$ colors, frequency range $[\sigma/3,2k\sigma]$ and time $O(D+\Delta+\lg^5\lg n)$, and (3) a coloring with exact $(\Delta+1)$-color palette, frequency range $[\sigma/2, 2\sigma+ \lceil\lg(\Delta+1)\rceil \lceil(1+\varepsilon)\sigma\rceil]$ and time $O((D+\Delta)\lg n\lg\Delta)$.
These \CONGEST\ running times thus retain an additive $D+\Delta$ term.
The governing question here is which guarantees survive when this global coordination cost is avoided.

The first procedure, termed \SamplerBalance, invokes the distributed coloring
sampler of Feng, Hayes, and Yin~\cite{FengHayesYin2018} as a black box.  With
$\numclr=\lceil(2+\delta)\Delta\rceil$, it samples a coloring whose probability law is close to the uniform distribution on proper colorings.  The contribution lies in the frequency analysis: a Dobrushin/self-bounding argument proves that all
color frequencies concentrate simultaneously under the target distribution,
and total-variation distance transfers this event to the sampler output.
Hence the procedure controls a global statistic without computing it.

\begin{thmA}
For every fixed $\delta>0$ and every $\lambda\geq1$, procedure \SamplerBalance\
computes, with probability at least $1-n^{-\lambda}$, a proper coloring with
$\numclr=\lceil(2+\delta)\Delta\rceil$ colors satisfying
\[
\left|\freq(i)-\frac{n}{\numclr}\right|
=O_\delta\!\left(\sqrt{(\lambda+1)\frac{n}{\numclr}\lg n}+(\lambda+1)\lg n\right)
\qquad\text{simultaneously for all }i\in[\numclr].
\]
It runs in $O_\delta((\lambda+1)\lg n)$ \CONGEST\ rounds, with no dependence on
the diameter $D$.  (See Theorem~\ref{thm:sampler-balance}.)
\end{thmA}

The second procedure, termed \SlackGreedy, uses an elementary parallel
random-greedy process.  Every uncolored vertex proposes a locally available
color, and a proposal is accepted when it has no simultaneous conflict. The excess $s=\numclr-\Delta$ is termed the \emph{palette slack}.  The same slack that pays for local progress also limits the proposal load of each color. This yields an upper-balanced coloring for every $\numclr>\Delta$, including a palette arbitrarily close to $\Delta+1$.  
The price is that our analysis provides no lower-frequency guarantee;
whether the algorithm itself satisfies one is an open question (see
Remark~\ref{rem:balance-expectation} and Section~\ref{sec:discussion}).

\begin{thmB}
For every integer $\numclr>\Delta$ and every $\lambda\geq1$, procedure
\SlackGreedy\ computes, with probability at least $1-n^{-\lambda}$, a proper
coloring using at most $\numclr$ colors in which every used color has
frequency at most
\[
\min\!\left\{n,\;
O\!\left((\lambda+1)\left(\frac{n\lg n}{\numclr-\Delta}+\lg n\right)\right)
\right\}.
\]
In particular, with $\numclr=\Delta+\lceil(\Delta+1)/\lceil\ln n\rceil\rceil$
colors, the cap is $\min\{n,\,O((\lambda+1)(\sigma\lg^2 n+\lg n))\}$, where
$\sigma=n/(\Delta+1)$.  It runs in $O((\lambda+1)\lg n)$ \CONGEST\ rounds,
with no dependence on the diameter $D$.
(See Theorem~\ref{thm:slack-greedy} and
Corollaries~\ref{cor:slack-greedy-fractional}
and~\ref{cor:slack-greedy-near-exact}.)
\end{thmB}

Table~\ref{tab:results} summarizes the two tradeoffs.  The exact frequency
parameters are defined in Theorems~\ref{thm:sampler-balance}
and~\ref{thm:slack-greedy}.

\begin{table}[t]
\centering
\footnotesize
\caption{Summary of the two procedures, preceded, for comparison, by
the three tradeoff points (1)--(3) of the companion
work~\cite{NirPeleg2026NearEquitable} listed above (whose running times carry
an additive diameter term)}.  The failure probability is at most
$n^{-\lambda}$.
\label{tab:results}
\begin{tabularx}{\textwidth}{@{}L{0.16\textwidth}L{0.18\textwidth}Y
 L{0.18\textwidth}L{0.20\textwidth}@{}}
\toprule
Procedure & Palette & Frequency guarantee & \CONGEST\ rounds
& Sequential time\\
\midrule
\cite{NirPeleg2026NearEquitable}\,(1)
& $2(\Delta+1)$
& Two-sided range $[\sigma/2,\,\sigma]$
& $O(\lg\Delta\,(D+\Delta+\lg^5\lg n))$
& --\\
\cite{NirPeleg2026NearEquitable}\,(2)
& $(1+1/k)(\Delta+1)$, $k\ge1$
& Two-sided range $[\sigma/3,\,2k\sigma]$
& $O(D+\Delta+\lg^5\lg n)$
& --\\
\cite{NirPeleg2026NearEquitable}\,(3)
& $\Delta+1$
& Two-sided range $[\sigma/2,\;2\sigma+\lceil\lg(\Delta+1)\rceil\lceil(1+\varepsilon)\sigma\rceil]$
& $O((D+\Delta)\lg n\,\lg\Delta)$
& --\\
\midrule
\SamplerBalance
& $\lceil(2+\delta)\Delta\rceil$
& Two-sided additive balance; relative balance under
  Condition~\eqref{eq:sampler-load}
& $O_\delta((\lambda+1)\lg n)$
& $O_\delta((m+n)(\lambda+1)\lg n)$\\[2mm]
\SlackGreedy
& Every $\numclr>\Delta$
& One-sided cap $\Psi_{p,\numclr}$; for the near-exact palette, see
  Corollary~\ref{cor:slack-greedy-near-exact}
& $O((\lambda+1)\lg n)$
& $O((m+n)(\lambda+1)\lg n)$ (expected)\\
\bottomrule
\end{tabularx}
\end{table}

The near-exact-palette endpoint of \SlackGreedy\ deserves one qualification.
Its exact cap always includes the deterministic minimum with $n$.  For fixed
$\lambda$, the asymptotic expression improves on the trivial cap only when
$\Delta$ is larger than a sufficiently large constant multiple of
$\lg^2 n$, and it is $o(n)$ when $\Delta/\lg^2 n$ tends to infinity.  Thus
the endpoint is most informative in the moderate- and high-degree regimes.

Two robustness remarks are in order.  First, exact knowledge of $\Delta$ is
not needed: if the vertices are given a common upper bound
$\hat\Delta\geq\Delta$ and the parameters below are computed from
$\hat\Delta$, then all guarantees continue to hold with $\Delta$ replaced by
$\hat\Delta$ throughout.  (For the first algorithm this holds because
$\numclr\geq(2+\delta)\hat\Delta\geq(2+\delta)\Delta$ still meets the
sampler's threshold; for the second, because the true conflict and proposal
probabilities are only smaller than the ones computed from $\hat\Delta$.)
Second, neither algorithm ever transmits a vertex identifier: messages carry
only colors, proposals, and activation bits.  Hence both algorithms run
unchanged in anonymous networks with port numbering, given common knowledge
of the parameters.

%%%%%%%%%%%%
\subsection{Relation to Previous Work}

Equitable coloring is introduced by Meyer~\cite{Meyer1973}; the existence
theory culminates in the Hajnal--Szemer\'edi theorem and its algorithmic
versions~\cite{HajnalSzemeredi1970,KiersteadEtAl2010}.  Parallel balanced
coloring is also studied as a practical load-balancing problem on shared-memory
and manycore machines~\cite{LuEtAl2017}.  Those algorithms and objectives
differ from the bandwidth-limited message-passing guarantees considered here.

The local symmetry-breaking literature supplies fast proper colorings but
does not ordinarily constrain global color-class sizes.  In particular, the
elementary randomized $(\Delta+1)$-coloring process is analyzed by
Johansson~\cite{Johansson1999}, and faster \LOCAL-model algorithms are known
for $(\Delta+1)$-coloring~\cite{HarrisSchneiderSu2016}.
In the \CONGEST\ model itself, Halld\'orsson, Kuhn, Maus, and Tonoyan
solve $(\Delta+1)$-list coloring in $O(\lg^5\lg n)$ rounds
w.h.p.~\cite{HalldorssonKMT2021}.  Measured against this properness-only
baseline, the $O((\lambda+1)\lg n)$ running times obtained here are the price
of the added frequency guarantees.
Distributed defective
and frugal colorings impose local restrictions on monochromatic degrees or on
repeated colors in neighborhoods~\cite{ChungPettieSu2014,FischerGhaffari2017};
these restrictions are different from the global frequency bounds studied
herein.

For the first result, Feng, Hayes, and Yin prove rapid mixing of the Lazy Local Metropolis chain above the $(2+\delta)\Delta$ threshold ~\cite{FengHayesYin2018}.  Fischer and Ghaffari independently obtain the same threshold using Local Glauber dynamics~\cite{FischerGhaffari2018}.  The \SamplerBalance\ procedure invokes the former chain; the latter result provides independent context for the threshold.
Below the $2\Delta$ threshold, the best known distributed sampler is
the distributed flip dynamics of Carlson, Frishberg, and
Vigoda~\cite{CarlsonFrishbergVigoda2023}, which handles
$\numclr>(11/6-\delta_0)\Delta$ colors for a fixed $\delta_0>0$ but requires
$O(n\lg n)$ rounds and leaves the stationary distribution unclear.  Thus the
palette size of \SamplerBalance\ reflects the current reach of
$O(\lg n)$-round distributed sampling, rather than a limitation of the
frequency analysis.
The frequency analysis uses Paulin's concentration inequality for self-bounding functions under Dobrushin dependence~\cite{Paulin2014}.

%No optimality claim is made for either palette--frequency tradeoff in this paper.

The rest of the paper is organized as follows.  Section~\ref{sec:prob-tools} states the probabilistic tools.  Section~\ref{sec:sampler-balance} presents the sampling-based algorithm and proves its two-sided frequency guarantee. Section~\ref{sec:slack-greedy} presents the coordination-free greedy algorithm and derives its palette--frequency tradeoffs.  Section~\ref{sec:discussion} discusses the limitations of the two approaches.

%%%%%%%%%%%%
\section{Probabilistic Tools}
\label{sec:prob-tools}

This section reviews the two concentration mechanisms used in the paper.
The first concerns dependent random vectors satisfying a Dobrushin
contraction condition.  The second controls an adaptively exposed sum of
Bernoulli variables.

%%%%%%%%%%%%
\subsection{Total-Variation Distance}

For two probability distributions $\nu$ and $\nu'$ on a finite space
$\Omega$, their \emph{total-variation distance} is
\begin{equation}
\label{eq:total-variation}
\TVD(\nu,\nu')
=\frac12\sum_{\omega\in\Omega}|\nu(\omega)-\nu'(\omega)|
=\max_{\cE\subseteq\Omega}|\nu(\cE)-\nu'(\cE)|.
\end{equation}
The following immediate consequence is the only %transfer 
property used hereafter.

\begin{fact}
%[Event comparison]
\label{fact:event-comparison}
If an event $\cE$ satisfies $\Prob_{\nu'}[\cE]\geq1-\beta$, then
$\Prob_{\nu}[\cE]\geq1-\beta-\TVD(\nu,\nu')$.
\end{fact}

%%%%%%%%%%%%
\subsection{Dobrushin Dependence and Self-Bounding Concentration}

Let $X=(X_v)_{v\in V}$ be a random vector on a finite product space.  Denote
by $X_{-v}=(X_u)_{u\ne v}$ the vector obtained by deleting coordinate $v$.
For a random variable $Y$, the notation $\mathcal L(Y)$ denotes its probability law.

For every vertex $v$, fix a version of the conditional kernel of $X_v$
given $X_{-v}$.  A nonnegative matrix
$A=(A_{vu})_{v,u\in V}$ with zero diagonal is termed a
\emph{Dobrushin interdependence matrix} for $X$ if, for every $v$ and every
two boundary configurations $x_{-v}$ and $y_{-v}$ in the domain of these
kernels,
\begin{equation}
%&
\TVD\!\left(
 \mathcal L(X_v\mid X_{-v}=x_{-v}),
 \mathcal L(X_v\mid X_{-v}=y_{-v})
\right)
%\nonumber\\
%&\hspace{35mm}
~\leq~
\sum_{u\ne v}A_{vu}\mathbf 1\{x_u\ne y_u\}.
\label{eq:dobrushin-definition}
\end{equation}
Here $\mathbf 1\{\mathcal S\}$ is the \emph{indicator} of a statement
$\mathcal S$, namely, it equals $1$ when $\mathcal S$ holds and $0$
otherwise.  Thus $A_{vu}$ bounds the effect of coordinate $u$ on the
conditional distribution at coordinate $v$.  The qualification concerning a
fixed kernel matters only on boundary configurations of probability zero,
where conditioning alone does not determine a unique distribution.

We use the maximum row and column sums
\[
\|A\|_\infty=\max_v\sum_u A_{vu},
\qquad
\|A\|_1=\max_u\sum_v A_{vu}.
\]
The inequality $\|A\|_1<1$ is the Dobrushin contraction condition employed
below.

A nonnegative function $g$ on the same product space is termed
\emph{$(1,0)$-$*$-self-bounding} if there are coordinate coefficients
$a_v(x)\in[0,1]$ such that, for every two configurations $x$ and $y$,
\begin{equation}
\label{eq:self-bounding-definition}
g(x)-g(y)\leq
\sum_{\substack{v\in V\\x_v\ne y_v}}a_v(x),
\qquad
\sum_{v\in V}a_v(x)\leq g(x).
\end{equation}
Intuitively, each changed coordinate can account for at most one unit of
decrease in $g$, and the total available charge is bounded by the current
value of $g$.

We use the following specialization of Paulin's concentration
inequality~\cite
%[Corollary~3.2]
{Paulin2014}.

% Environment changed from "corollary" to "fact" (imported result) -- Claude
\begin{fact}
{\bf \cite[Corollary~3.2]{Paulin2014}}
%[Self-bounding concentration under Dobrushin dependence]
\label{cor:paulin-self-bounding}
Let $A$ be a Dobrushin interdependence matrix for $X$ such that
$\|A\|_1<1$ and $\|A\|_\infty\leq1$.  Let $g$ be a nonnegative
$(1,0)$-$*$-self-bounding function, and put
$\Gamma=1-\|A\|_1$.  Then, for every $t\geq0$,
\begin{align}
\Prob[g(X)\geq\Exp[g(X)]+t]
&\leq
\exp\!\left(-\frac{\Gamma t^2}
{2(\Exp[g(X)]+t)}\right),
\label{eq:paulin-upper}\\
\Prob[g(X)\leq\Exp[g(X)]-t]
&\leq
\exp\!\left(-\frac{\Gamma t^2}
{8\Exp[g(X)]}\right).
\label{eq:paulin-lower}
\end{align}
\end{fact}

The lower-tail inequality is Part~(3) of the cited corollary.  A
$(1,0)$-$*$-self-bounding function is also weakly
$(1,0)$-$*$-self-bounding, and its value changes by at most one when a single
coordinate is changed.  Moreover, the constant $a_c$ in
\cite[Corollary~3.2]{Paulin2014} satisfies $a_c<0.286$.  Hence the additional
condition $1\geq a_c(1-\|A\|_1)$ required there holds automatically.

%%%%%%%%%%%%
\subsection{Adaptive Bernoulli Sums}

The frequency analysis of \SlackGreedy\ uses a conditional form of the Bernstein--Chernoff bound.  The following statement records the standard moment-generating-function estimate that permits adaptive trials.
It follows from the stochastic domination of $(Y_j)_j$ by a sequence of independent Bernoulli($q_j$) trials, to which the usual Bernstein--Chernoff bound applies; see Doerr~\cite[Section~1.10.2]{Doerr2020} for the domination step and \cite[Chapter~1]{DubhashiPanconesi2009} for the Chernoff computation. Alternatively, it is a special case of Freedman's martingale
inequality~\cite{Freedman1975}.

\begin{fact}
%[Adaptive Bernstein--Chernoff bound]
\label{fact:adaptive-bernoulli}
Let $Y_1,\ldots,Y_N\in\{0,1\}$ be adapted to a filtration
$\cF_0,\ldots,\cF_N$.  Suppose that deterministic numbers
$q_1,\ldots,q_N\in[0,1]$ satisfy
\[
\Prob[Y_j=1\mid\cF_{j-1}]\leq q_j
\qquad\text{almost surely for every }j\in[N].
\]
Denote $\mu=\sum_{j=1}^Nq_j$.  Then, for every $\vartheta\geq0$ and every
$x\geq0$,
\begin{align}
\Exp\!\left[\exp\!\left(\vartheta\sum_{j=1}^NY_j\right)\right]
&\leq \exp\bigl((e^{\vartheta}-1)\mu\bigr),
\label{eq:adaptive-mgf}\\
\Prob\!\left[\sum_{j=1}^NY_j>
\mu+\sqrt{2\mu x}+\frac{2x}{3}\right]
&\leq e^{-x}.
\label{eq:adaptive-bernstein}
\end{align}
\end{fact}

Note that Fact~\ref{fact:adaptive-bernoulli} requires no independence between the
variables $Y_j$.  It suffices to bound each success probability conditioned
on the information revealed earlier.

%%%%%%%%%%%%
\section{Near-Equitable Coloring by Local Sampling}
\label{sec:sampler-balance}

The first tradeoff point retains a two-sided frequency guarantee and removes
global coordination by enlarging the palette to about $2\Delta$.  Its
algorithmic ingredient is the Lazy Local Metropolis sampler of Feng, Hayes,
and Yin~\cite{FengHayesYin2018}.  We first state the transition rule and the
black-box mixing guarantee precisely.  We then prove frequency concentration
under the uniform proper-coloring measure and transfer the resulting event to
the sampler output.

%%%%%%%%%%%%
\subsection{The Lazy Local Metropolis Sampler}

For an integer $\numclr\geq1$, let
$\Omega_{\numclr}=[\numclr]^V$.  Denote by $\pi_{\numclr}$ the uniform
distribution over the proper $\numclr$-colorings of $G$, viewed as a
distribution on $\Omega_{\numclr}$.

We now describe one step of the Lazy Local Metropolis chain.  Let
$X_t\in\Omega_{\numclr}$ be the current coloring and put
$p_\delta=\min\{\delta/3,1/2\}$.  Every vertex becomes active independently
with probability $p_\delta$.  Each active vertex $v$ chooses a proposal
$c_t(v)$ uniformly from $[\numclr]$ and sends the pair
$(X_t(v),c_t(v))$ to its neighbors.

An edge $uv$ with two active endpoints is said to \emph{pass} if
\begin{equation}
\label{eq:active-edge-passes}
c_t(u)\ne c_t(v),\qquad
c_t(u)\ne X_t(v),\qquad
X_t(u)\ne c_t(v).
\end{equation}
An edge $uv$ with active endpoint $v$ and inactive endpoint $u$ passes if
\begin{equation}
\label{eq:boundary-edge-passes}
c_t(v)\ne X_t(u).
\end{equation}
An active vertex adopts its proposal only if every incident edge passes.
Every other vertex retains its current color.  Note that
Eq.~\eqref{eq:active-edge-passes} contains three distinct checks; all three
are part of the chain analyzed in~\cite{FengHayesYin2018}.

The chain may be initialized at an arbitrary $\numclr$-coloring, which need
not be proper.  Every step maps a total coloring to a total coloring, although
intermediate colorings need not be proper.  The sampling property required
here is the following result of Feng, Hayes, and
Yin~\cite{FengHayesYin2018}.

\begin{fact}
%[Distributed coloring sampler]
{\bf \cite{FengHayesYin2018}}
\label{fact:distributed-coloring-sampler}
For every fixed $\delta>0$, there is a constant $C_\delta>0$ with the
following property.  Let $\numclr\geq(2+\delta)\Delta$ and
$0<\zeta<1$.  Starting from any $X_0\in\Omega_{\numclr}$, after at most
$C_\delta\bigl(\ln n+\ln(1/\zeta)\bigr)$ rounds, the Lazy Local Metropolis algorithm returns a coloring $\widehat X$
whose distribution $\nu$ satisfies
$\TVD(\nu,\pi_{\numclr})\leq\zeta$.  Every message contains $O(\lg n)$ bits.
In particular,
$\Prob[\widehat X\text{ is improper}]\leq\zeta$.
\end{fact}

Fischer and Ghaffari~\cite{FischerGhaffari2018} obtain the same $\numclr>2\Delta$ threshold for a different chain, namely, Local Glauber
dynamics.
%Fact~\ref{fact:distributed-coloring-sampler} invokes only the Lazy Local Metropolis result of Feng, Hayes, and Yin.

%%%%%%%%%%%%
\subsection{Procedure \SamplerBalance}

Fix $\delta>0$ and put
\[
\numclr=\lceil(2+\delta)\Delta\rceil,
\qquad
\sigma_{\numclr}=\frac{n}{\numclr}.
\]
Procedure \SamplerBalance\ starts from an arbitrary total
$\numclr$-coloring and executes the sampler with target total-variation error
$\zeta=n^{-(\lambda+4)}$.  Its only persistent local state is the current
color.  At the end, one additional neighbor round verifies properness.

The complete procedure is given in Algorithm~\ref{alg:sampler-balance}.

\bigskip
%%%%%%%%%%%%%%%%%%%%%%%%%%%
\begin{algorithm}[H]
\caption{Procedure \SamplerBalance}
\label{alg:sampler-balance}
\KwIn{Graph $G=(V,E)$, constants $\delta>0$ and $\lambda\geq1$.}
\KwOut{A proper $\numclr$-coloring, or local \textsf{failure}.}
$\numclr\gets\lceil(2+\delta)\Delta\rceil$ and
$\zeta\gets n^{-(\lambda+4)}$\;
Start from an arbitrary total $\numclr$-coloring of $V$\;
Run the Lazy Local Metropolis algorithm for
$C_\delta(\ln n+\ln(1/\zeta))$ rounds, where $C_\delta$ is the constant from
Fact~\ref{fact:distributed-coloring-sampler}\;
Every vertex sends its final color to its neighbors\;
Every endpoint of a monochromatic edge declares local \textsf{failure};
all other vertices output their final colors\;
\end{algorithm}

\bigskip
If a failure is declared anywhere, the distributed execution is considered
unsuccessful.  The procedure requires no global failure-detection mechanism:
the theorem below concerns the event that no vertex declares failure.

%%%%%%%%%%%%
\subsection{Frequency Concentration under the Uniform Measure}

We next analyze a proper coloring drawn exactly from $\pi_{\numclr}$.  Put
\begin{equation}
\label{eq:dobrushin-parameters}
\alpha=\frac{\Delta}{\numclr-\Delta},
\qquad
\gamma=1-\alpha.
\end{equation}
The choice of $\numclr$ gives
\[
\alpha\leq\frac{1}{1+\delta}<1,
\qquad
\gamma\geq\frac{\delta}{1+\delta}>0.
\]

The main analytic ingredient is the following concentration statement.

\begin{lemma}
%[Frequency concentration]
\label{lem:uniform-color-frequency}
Let $X\sim\pi_{\numclr}$.  For every $i\in[\numclr]$, define
\[
\Phi_i(x)=|\{v\in V\mid x_v=i\}|,
\qquad x\in\Omega_{\numclr}.
\]
Then $\Exp[\Phi_i(X)]=\sigma_{\numclr}$, and, for every $t\geq0$,
\begin{align}
\Prob[\Phi_i(X)\geq\sigma_{\numclr}+t]
&\leq
\exp\!\left(-\frac{\gamma t^2}{2(\sigma_{\numclr}+t)}\right),
\label{eq:sampler-upper-tail}\\
\Prob[\Phi_i(X)\leq\sigma_{\numclr}-t]
&\leq
\exp\!\left(-\frac{\gamma t^2}{8\sigma_{\numclr}}\right).
\label{eq:sampler-lower-tail}
\end{align}
\end{lemma}

\begin{proof}
Fix a color $i$.  We first establish a Dobrushin interdependence bound.  For
every boundary configuration $x_{-v}$, including configurations of
probability zero, define the conditional kernel at $v$ to be uniform over
\begin{equation}
\label{eq:available-kernel}
\cP_v(x)=[\numclr]\setminus\{x_u\mid u\in N(v)\}.
\end{equation}
This palette is nonempty since $\numclr>\deg(v)$.  On every feasible boundary
configuration, Eq.~\eqref{eq:available-kernel} is exactly the conditional law
of $X_v$ under $\pi_{\numclr}$.

Consider first two boundary configurations that differ only at a coordinate
$u$.  If $uv\notin E$, their kernels at $v$ are identical.  If $uv\in E$,
the two available palettes differ by at most one deletion and one insertion.
Both palettes have size at least $\numclr-\deg(v)$.  A direct comparison of
the two uniform distributions therefore gives total-variation distance at
most $1/(\numclr-\deg(v))$.
Indeed, if the two palettes coincide, the distance is zero.  If they
are $Q\cup\{a\}$ and $Q\cup\{b\}$ with $a\ne b$ and $|Q|=k-1$ (one deletion
and one insertion), the distance is $1/k$; if one is $Q$ and the other
$Q\cup\{b\}$ with $|Q|=k$, the distance is $1/(k+1)$.  In every case the
distance is at most one over the smaller palette size, which is at least
$\numclr-\deg(v)$.

Now consider arbitrary boundary configurations $x_{-v}$ and $y_{-v}$.  Change
their differing coordinates one at a time.  Applying the preceding
single-coordinate estimate at every step and the triangle inequality for
total-variation distance yields Eq.~\eqref{eq:dobrushin-definition} for the
matrix
\[
A_{vu}=
\begin{cases}
1/(\numclr-\deg(v)),&uv\in E,\\
0,&uv\notin E.
\end{cases}
\]
Its maximum row and column sums satisfy
\begin{align*}
\|A\|_\infty
&=\max_v\frac{\deg(v)}{\numclr-\deg(v)}
\leq\frac{\Delta}{\numclr-\Delta},\\
\|A\|_1
&=\max_u\sum_{v\in N(u)}\frac{1}{\numclr-\deg(v)}
\leq\frac{\Delta}{\numclr-\Delta}.
\end{align*}
Thus both norms are at most $\alpha<1$.

We next verify the self-bounding property.  For a configuration
$x=(x_v)_{v\in V}$, set $a_v(x)=\mathbf 1\{x_v=i\}$.  For every $x$ and $y$,
\[
\Phi_i(x)-\Phi_i(y)
\leq\sum_{\substack{v\in V\\x_v\ne y_v}}a_v(x),
\qquad
\sum_{v\in V}a_v(x)=\Phi_i(x).
\]
Indeed, every occurrence of color $i$ lost in passing from $x$ to $y$ is
charged to its changed coordinate.
In more detail, write $S_x=\{v\in V\mid x_v=i\}$, so that
$\Phi_i(x)=|S_x|$ and $a_v(x)=\mathbf 1\{v\in S_x\}$.  Then
$\Phi_i(x)-\Phi_i(y)\leq|S_x\setminus S_y|$, and every
$v\in S_x\setminus S_y$ satisfies $x_v=i\neq y_v$, so it contributes
$a_v(x)=1$ to the sum over the changed coordinates.  The second relation
holds with equality, since $\sum_v a_v(x)=|S_x|=\Phi_i(x)$.
Hence $\Phi_i$ is $(1,0)$-$*$-self-bounding.

Finally, color-name permutations preserve $\pi_{\numclr}$.  The random
variables $\Phi_1(X),\ldots,\Phi_{\numclr}(X)$ therefore have the same mean, and
their sum is always $n$.  It follows that
$\Exp[\Phi_i(X)]=n/\numclr=\sigma_{\numclr}$.  Applying
Fact~\ref{cor:paulin-self-bounding} with
$1-\|A\|_1\geq\gamma$ proves
Eqs.~\eqref{eq:sampler-upper-tail} and~\eqref{eq:sampler-lower-tail}.
\end{proof}

%\begin{remark}
%[Why self-bounding, rather than bounded differences]
%\label{rem:why-self-bounding}
We remark that changing one coordinate changes $\Phi_i$ by at most one, so a
bounded-differences inequality under the Dobrushin condition would bound the
deviation of $\Phi_i$ by $O(\sqrt{n\lg n})$ w.h.p., the variance proxy being
the number $n$ of coordinates.  That bound is informative only when
$\sqrt{n\lg n}\ll\sigma_{\numclr}$, i.e., when
$\Delta\ll\sqrt{n/\lg n}$.  The self-bounding form replaces the ambient
dimension $n$ by the mean $\sigma_{\numclr}$ in the variance proxy, which is
what extends the useful range of Theorem~\ref{thm:sampler-balance} to
$\Delta\ll n/\lg n$ (cf.\ Condition~\eqref{eq:sampler-load}).
%\end{remark}

To cover every color simultaneously, define
\begin{equation}
\label{eq:sampler-additive-error}
R_{\mathrm{samp}}=(\lambda+4)\ln n+\ln(2\numclr),
\qquad
B_{\numclr}=
\sqrt{\frac{8\sigma_{\numclr}R_{\mathrm{samp}}}{\gamma}}
+\frac{2R_{\mathrm{samp}}}{\gamma}.
\end{equation}

%%%%%%%%%%%%
\subsection{Main Sampling Guarantee}

We obtain the following diameter-independent two-sided guarantee.

\begin{theorem}
\label{thm:sampler-balance}
Fix $\delta>0$ and $\lambda\geq1$, let $n\geq2$ (as implied by the
standing assumption $\Delta\geq1$), and let
$\numclr=\lceil(2+\delta)\Delta\rceil$.

\smallskip
\noindent
\textnormal{(a)}
Algorithm~\ref{alg:sampler-balance} computes, with probability at least
$1-n^{-\lambda}$, a proper $B_{\numclr}$-additively balanced
$\numclr$-coloring.  More explicitly, every $i\in[\numclr]$ satisfies
\begin{equation}
\label{eq:sampler-explicit-bound}
\max\{0,\sigma_{\numclr}-B_{\numclr}\}
\leq\freq(i)\leq\sigma_{\numclr}+B_{\numclr}.
\end{equation}
In particular,
\begin{equation}
\label{eq:sampler-asymptotic-bound}
\left|\freq(i)-\frac{n}{\numclr}\right|
=O_\delta\!\left(
\sqrt{(\lambda+1)\frac{n}{\numclr}\lg n}
+(\lambda+1)\lg n
\right)
\end{equation}
simultaneously for all colors $i\in[\numclr]$.

\smallskip
\noindent
\textnormal{(b)}
Fix $0<\eta\leq1$.  If
\begin{equation}
\label{eq:sampler-load}
\sigma_{\numclr}\geq
\frac{8R_{\mathrm{samp}}}{\gamma\eta^2},
\end{equation}
then, with the same probability, the output is
$\eta$-relatively balanced.

\smallskip\noindent
The 
%expected
running times are $O_\delta((m+n)(\lambda+1)\lg n)$ sequentially and $O_\delta((\lambda+1)\lg n)$ in the \CONGEST\ model.
\end{theorem}

\begin{proof}
We start with Part~(a).  Let $X\sim\pi_{\numclr}$ and fix a color $i$.
Put $r=R_{\mathrm{samp}}/\gamma$.  By
Eq.~\eqref{eq:sampler-additive-error},
$B_{\numclr}=\sqrt{8\sigma_{\numclr}r}+2r$.  The lower-tail exponent in
Eq.~\eqref{eq:sampler-lower-tail} is at least
$\gamma B_{\numclr}^2/(8\sigma_{\numclr})\geq R_{\mathrm{samp}}$.
For the upper tail, a direct calculation gives
\begin{align*}
B_{\numclr}^2
-2r(\sigma_{\numclr}+B_{\numclr})
&=6\sigma_{\numclr}r
+2r\sqrt{8\sigma_{\numclr}r}
\geq0.
\end{align*}
Hence the exponent in Eq.~\eqref{eq:sampler-upper-tail} is also at least
$R_{\mathrm{samp}}$.  A union bound over the two tails and all
$\numclr$ colors shows that every color satisfies
Eq.~\eqref{eq:sampler-explicit-bound} with probability at least
\[
1-2\numclr e^{-R_{\mathrm{samp}}}
=1-n^{-(\lambda+4)}
\]
under $\pi_{\numclr}$.

Let $\cE_{\mathrm{add}}$ denote the explicit event in
Eq.~\eqref{eq:sampler-explicit-bound}.  This event is well-defined without
invoking asymptotic notation.  Let $\widehat X$ be the output of the sampler
and let $\nu=\mathcal L(\widehat X)$.  Fact
~\ref{fact:distributed-coloring-sampler} gives
$\TVD(\nu,\pi_{\numclr})\leq n^{-(\lambda+4)}$.  Since every coloring in the
support of $\pi_{\numclr}$ is proper, Fact~\ref{fact:event-comparison} gives
\[
\Prob_\nu[\widehat X\text{ is proper and }\cE_{\mathrm{add}}\text{ holds}]
\geq1-2n^{-(\lambda+4)}
\geq1-n^{-\lambda}\;\mbox{(using $n\geq2$)}.
\]
This proves the explicit claim.  Since
$\gamma\geq\delta/(1+\delta)$ and
$R_{\mathrm{samp}}=O_\delta((\lambda+1)\lg n)$,
Eq.~\eqref{eq:sampler-asymptotic-bound} follows deterministically from
Eq.~\eqref{eq:sampler-additive-error}.

Let us now prove Part~(b).  Apply
Eqs.~\eqref{eq:sampler-upper-tail} and~\eqref{eq:sampler-lower-tail} with
$t=\eta\sigma_{\numclr}$.  Since $0<\eta\leq1$, the upper-tail exponent is at
least $\gamma\eta^2\sigma_{\numclr}/4$, and the lower-tail exponent equals
$\gamma\eta^2\sigma_{\numclr}/8$.  Condition~\eqref{eq:sampler-load} makes
both exponents at least $R_{\mathrm{samp}}$.  The same union bound and
total-variation transfer prove the relative guarantee.

It remains to account for the running times.  Fact
~\ref{fact:distributed-coloring-sampler} gives $O_\delta((\lambda+1)\lg n)$ sampler rounds, and the final verification adds one round. A sequential simulation processes the $n$ activation/proposal choices and scans the $m$ edges in every round, yielding worst-case time $O_\delta((m+n)(\lambda+1)\lg n)$, since each proposal is drawn uniformly from $[\numclr]$ with no rejection step.  This completes the proof.
\end{proof}

Theorem~\ref{thm:sampler-balance} should be considered in three regimes. 
%Three regimes clarify when Theorem \ref{thm:sampler-balance} carries information. 
The upper bound in Eq.~\eqref{eq:sampler-explicit-bound} is
informative whenever $\sigma_{\numclr}+B_{\numclr}<n$.  The lower bound ---
and with it the two-sided additive guarantee --- is informative exactly when
$B_{\numclr}<\sigma_{\numclr}$, which holds if and only if
$\sigma_{\numclr}=\Omega_\delta((\lambda+1)\lg n)$; up to the constant, this
is Condition~\eqref{eq:sampler-load} with constant $\eta$.  Below this load,
$\max\{0,\sigma_{\numclr}-B_{\numclr}\}=0$ and the lower half of
Eq.~\eqref{eq:sampler-explicit-bound} degenerates to the trivial bound
$\freq(i)\geq0$.  Finally, Condition~\eqref{eq:sampler-load} upgrades the
additive guarantee to $\eta$-relative balance. Thus Condition~\eqref{eq:sampler-load} 
%should be read not as a technical convenience but as delineating, 
delineated, essentially, the regime in which a two-sided guarantee is available at all.

%\begin{remark}
%[Near-optimality of the additive bound]
%\label{rem:sampler-optimality}
We remark that under the target distribution itself, the deviation in
Eq.~\eqref{eq:sampler-asymptotic-bound} is essentially best possible for this
palette.  Consider, for instance, the graph consisting of $n/(\Delta+1)$
disjoint copies of $K_{\Delta+1}$.  Under $\pi_{\numclr}$ the cliques are
colored independently, so each frequency $\Phi_i(X)$ is a
$\mathrm{Binomial}(n/(\Delta+1),(\Delta+1)/\numclr)$ variable with mean
$\sigma_{\numclr}$ and variance $\Theta(\sigma_{\numclr})$; already a single
color deviates by $\Omega(\sqrt{\sigma_{\numclr}})$ with constant
probability, and a routine second-moment computation over the $\numclr$
colors raises this to $\Omega(\sqrt{\sigma_{\numclr}\lg\numclr})$ for the
maximal deviation (when $\lg\numclr=O(\sigma_{\numclr})$).  Hence the first
term of Eq.~\eqref{eq:sampler-asymptotic-bound} is necessary up to the
dependence on $\lambda$ and the gap between $\lg n$ and $\lg\numclr$.
%\end{remark}

%%%%%%%%%%%%
\section{Coordination-Free Coloring with Palette Slack}
\label{sec:slack-greedy}

The second tradeoff point brings the palette arbitrarily close to
$\Delta+1$ and retains a diameter-independent running time.  The price is a
one-sided guarantee: every used color class is capped, but small classes are
not prevented.  The algorithm requires only activations, locally available
colors, and neighbor conflict tests.

Intuitively, palette slack serves two roles.  It reduces the probability that
a proposal conflicts with a neighbor, and it reduces the maximum conditional
probability of proposing any fixed color.  More formally, these two effects
yield the progress and proposal-load bounds below.

%%%%%%%%%%%%
\subsection{Procedure \SlackGreedy}

Fix an integer $\numclr>\Delta$ whose color identifiers use $O(\lg n)$ bits,
and denote the palette slack by
\[
s=\numclr-\Delta.
\]
Set
\begin{equation}
\label{eq:slack-parameters}
p=\min\!\left\{1,\frac{\numclr}{2\Delta}\right\},
\qquad
\theta=1-p+\frac{p^2\Delta}{\numclr},
\end{equation}
and define the predetermined iteration count
\begin{equation}
\label{eq:slack-iteration-count}
T=T_{p,\numclr}=
\left\lceil
\frac{(\lambda+3)\ln n}{\ln(1/\theta)}
\right\rceil.
\end{equation}
The choice of $p$ in Eq.~\eqref{eq:slack-parameters} is the minimizer
of $\theta=1-p+p^2\Delta/\numclr$ over $p\in[0,1]$, balancing the activation
rate against the conflict rate.
If $\numclr<2\Delta$, then
$p=\numclr/(2\Delta)$ and
$\theta=1-\numclr/(4\Delta)<3/4$.  If
$\numclr\geq2\Delta$, then $p=1$ and
$\theta=\Delta/\numclr\leq1/2$.  Consequently,
\begin{equation}
\label{eq:slack-T-asymptotic}
T=O((\lambda+1)\lg n)
\qquad\text{for every }\numclr>\Delta.
\end{equation}

At the beginning of an iteration, an uncolored vertex deletes the colors
already used by its colored neighbors.  Active uncolored vertices then
propose colors in parallel.  A proposal is accepted exactly when no active
uncolored neighbor makes the same proposal.  All acceptances within an
iteration are simultaneous.  The number of iterations is fixed in advance,
so no global termination test or global color counter is required.

The formal procedure is given in Algorithm~\ref{alg:slack-greedy}.

\bigskip
%%%%%%%%%%%%%%%%%%%%%%%%%%%%
\begin{algorithm}[H]
\caption{Procedure \SlackGreedy}
\label{alg:slack-greedy}
\KwIn{Graph $G=(V,E)$, palette size $\numclr>\Delta$, and
failure exponent $\lambda\geq1$.}
\KwOut{A proper coloring, or local \textsf{failure}.}
Compute $p$ and $T$ from Eqs.~\eqref{eq:slack-parameters}
and~\eqref{eq:slack-iteration-count}; set $U\gets V$\;
\For{$t=1,\ldots,T$}{
  Every $v\in U$ becomes active independently with probability $p$\;
  \ForEach{active $v\in U$ in parallel}{
    $\cP_t(v)\gets
    [\numclr]\setminus\{\clr(u)\mid u\in N(v)\setminus U\}$\;
    Choose $c_t(v)\in\cP_t(v)$ uniformly and send it to every neighbor\;
  }
  Every active $v\in U$ whose proposal differs from all proposals of active
  vertices in $N(v)\cap U$ sets $\clr(v)\gets c_t(v)$\;
  Every newly colored vertex sends its color to its neighbors and removes
  itself from $U$\;
}
Every vertex in $U$ declares local \textsf{failure}; every other vertex
outputs its color\;
\end{algorithm}

%%%%%%%%%%%%
\subsection{Progress and Termination}

We first establish the properness invariant and the one-iteration progress
bound.

\begin{lemma}
%[One-iteration progress]
\label{lem:slack-greedy-progress}
At every iteration, the current partial coloring is proper.  Moreover, for
every uncolored vertex $v$ and every execution history up to the beginning of
the iteration,
\[
\Prob[v\text{ remains uncolored after the iteration}\mid\text{history}]
\leq\theta.
\]
\end{lemma}

\begin{proof}
Let $U$ be the uncolored set at the beginning of the iteration.  Denote
$\deg_U(v)=|N(v)\cap U|$ and $a_v=|\cP_t(v)|$.  At most
$\Delta-\deg_U(v)$ colored neighbors exclude colors.  Hence
\begin{equation}
\label{eq:slack-available}
a_v\geq\numclr-(\Delta-\deg_U(v))
=s+\deg_U(v).
\end{equation}
An active vertex proposes a color unused by its colored neighbors.  If two
adjacent active vertices make the same proposal, neither is accepted.
Therefore simultaneous acceptance preserves properness.

Condition on the history and on $v$ being active.  For an uncolored neighbor
$u$ of $v$, let $\cE_{\mathrm{confl}}(u)$ be the event that $u$ is active and
$c_t(u)=c_t(v)$.  Fresh activations and proposals are mutually independent,
so Eq.~\eqref{eq:slack-available} gives
% Environment changed from equation to align to restore the first step -- Claude
\begin{align}
\Prob[\cE_{\mathrm{confl}}(u)\mid
\text{history},v\text{ active}]
&~=~ 
\Prob[u\text{ is active}]\cdot\!\!\!\sum_{i\in \cP_t(u)\cap \cP_t(v)}\!\!\!
\Prob[c_t(v)=i] \cdot \Prob[c_t(u)=i]
\nonumber\\
&~=~ p\sum_{i\in\cP_t(u)\cap\cP_t(v)}\frac{1}{a_ua_v}
~\leq~\frac{p}{a_v}
~\leq~ \frac{p}{s+\deg_U(v)}.
\label{eq:single-neighbor-conflict}
\end{align}
A union bound over the uncolored neighbors yields
\begin{equation*}
\Prob[v\text{ has a conflict}\mid
\text{history},v\text{ active}]
~\leq~
\frac{p\deg_U(v)}{s+\deg_U(v)}
~\leq~ \frac{p\Delta}{s+\Delta}
~=~ \frac{p\Delta}{\numclr}~.
\end{equation*}
The second inequality holds because $x/(s+x)$ is increasing for $x\geq0$.
Thus $v$ is colored with 
%conditional 
%Previous version was "unconditional"
probability at least
$p(1-p\Delta/\numclr)$, and the probability that it remains uncolored
is at most $1-p+p^2\Delta/\numclr=\theta$.
\end{proof}

The progress bound can be multiplied across iterations without an
independence assumption: one conditions successively on the complete history
at the beginning of each iteration.  It follows that a fixed vertex remains
uncolored after all $T$ iterations with probability at most
$\theta^T\leq n^{-(\lambda+3)}$.  A union bound now gives the following
termination guarantee.

\begin{corollary}
%[Termination]
\label{cor:slack-greedy-termination}
Algorithm~\ref{alg:slack-greedy} colors every vertex with probability at
least $1-n^{-(\lambda+2)}$.
\end{corollary}

%%%%%%%%%%%%
\subsection{Frequency Analysis}

An accepted proposal of color $i$ contributes one vertex to $\freq(i)$,
whereas a rejected proposal contributes nothing.  Hence the total number of
proposals of $i$ is an upper bound on its final frequency.
Define
\begin{equation}
\label{eq:slack-cap-parameters}
\mu_{p,\numclr}=\frac{pnT}{\numclr-\Delta},
\qquad
R_{\mathrm{prop}}=(\lambda+3)\ln n+\ln\numclr,
\end{equation}
and put
\begin{equation}
\label{eq:slack-frequency-cap}
\Psi_{p,\numclr}=
\min\!\left\{
n,
\left\lceil
\mu_{p,\numclr}
+\sqrt{2\mu_{p,\numclr}R_{\mathrm{prop}}}
+\frac{2R_{\mathrm{prop}}}{3}
\right\rceil
\right\}.
\end{equation}

The proposal process gives the following simultaneous cap.

\begin{lemma}
%[Frequency cap]
\label{lem:slack-greedy-frequency}
With probability at least $1-n^{-(\lambda+3)}$, every color used by
Algorithm~\ref{alg:slack-greedy} has frequency at most
$\Psi_{p,\numclr}$.
\end{lemma}

\begin{proof}
Fix a color $i$, and let $Z_i$ denote its total number of proposals, including
rejected proposals.  We define the filtration used to expose these proposals.
Order the $nT$ potential vertex--iteration trials lexicographically by
iteration and by an arbitrary fixed vertex order.  At the start of an
iteration, the current uncolored set and every available palette are already
determined by earlier iterations.  Within that iteration, expose the fresh
activation and, when active, the proposal of each vertex in the fixed order.
Acceptance decisions are determined only after all proposals of the
iteration have been exposed, and are added to the history before the next
iteration begins.  A trial corresponding to an already colored vertex is
padded with the value zero.

Let $Y_j$ indicate that trial $j$ proposes color $i$, and let
$\cF_j$ contain the information exposed through that trial together with all
acceptance decisions of completed iterations.  Conditioned on
$\cF_{j-1}$, the probability that trial $j$, corresponding to a vertex $v$,
proposes $i$ is either zero or
$p/|\cP_t(v)| \leq p/(\numclr-\Delta)$.
Thus Fact~\ref{fact:adaptive-bernoulli} applies with
$q_j=p/(\numclr-\Delta)$ for all $nT$ trials and total budget
$\mu_{p,\numclr}$.  It follows that
\[
\Prob\!\left[
Z_i>
\mu_{p,\numclr}
+\sqrt{2\mu_{p,\numclr}R_{\mathrm{prop}}}
+\frac{2R_{\mathrm{prop}}}{3}
\right]
\leq e^{-R_{\mathrm{prop}}}.
\]

Every accepted proposal of $i$ is counted by $Z_i$, so
$\freq(i)\leq Z_i$.  We also have the deterministic bound
$\freq(i)\leq n$.  Equation~\eqref{eq:slack-frequency-cap} and a union bound
over all palette colors give
\[
\Prob[\text{some color }i\text{ satisfies }\freq(i)>\Psi_{p,\numclr}]
\leq\numclr e^{-R_{\mathrm{prop}}}
=n^{-(\lambda+3)}.
\]
The lemma follows.
\end{proof}

\begin{remark}[The cost of a deterministic proposal budget]
\label{rem:budget}
The budget $\mu_{p,\numclr}$ charges every vertex a proposal at every one of
the $T$ iterations.  In expectation, far fewer proposals occur: by
Lemma~\ref{lem:slack-greedy-progress}, a vertex is still uncolored at the
start of iteration $t$ with probability at most $\theta^{t-1}$, so the total
number $Z_i$ of proposals of a fixed color $i$ satisfies
\[
\Exp[Z_i]
\;\leq\;
\frac{p}{\numclr-\Delta}\sum_{t\geq1}n\,\theta^{t-1}
\;=\;
\frac{pn}{(\numclr-\Delta)(1-\theta)}
\;\leq\;
\frac{2n}{\numclr-\Delta},
\]
since $1-\theta=p(1-p\Delta/\numclr)\geq p/2$.  Thus
$\Exp[Z_i]=O(n/(\numclr-\Delta))$, independently of $T$ --- a factor
$\Theta(pT)$ below $\mu_{p,\numclr}$.  The factor $T$ in
$\Psi_{p,\numclr}$ is therefore an artifact of the deterministic budget:
Fact~\ref{fact:adaptive-bernoulli} requires deterministic bounds $q_j$,
whereas the geometric decay of the uncolored set is a random event that the
conditioning has already exposed, so the analysis cannot take credit for it.
Converting this expectation benchmark into a high-probability cap --- say,
via martingale inequalities with random predictable variance, combined with
concentration of $\sum_t|U_t|$ --- would improve the caps of
Corollaries~\ref{cor:slack-greedy-fractional}
and~\ref{cor:slack-greedy-near-exact} by a $\Theta(\lg n)$ factor.  We leave
this as an open question.
\end{remark}

%%%%%%%%%%%%
\subsection{Main Guarantee and Tradeoff Points}

Combining progress with the proposal cap yields the general guarantee.

\begin{theorem}
\label{thm:slack-greedy}
Let $n\geq2$, $\lambda\geq1$, and let $\numclr>\Delta$ be an integer whose
color identifiers use $O(\lg n)$ bits.  Algorithm
~\ref{alg:slack-greedy} computes, with probability at least
$1-n^{-\lambda}$, a proper $\Psi_{p,\numclr}$-upper-balanced coloring using at
most $\numclr$ colors.  Its expected sequential running time is
$O((m+n)T)$, and its running time in the \CONGEST\ model is $O(T)$, where $T=O((\lambda+1)\lg n)$.
\end{theorem}

\begin{proof}
Corollary~\ref{cor:slack-greedy-termination} gives proper completion with
failure probability at most $n^{-(\lambda+2)}$, and
Lemma~\ref{lem:slack-greedy-frequency} gives the simultaneous cap with failure
probability at most $n^{-(\lambda+3)}$.  For $n\geq2$, their sum is at most
$n^{-\lambda}$.

Every iteration uses one neighbor round for proposals and one for acceptance
announcements.  Each message contains an activation bit or a color identifier
and therefore uses $O(\lg n)$ bits.  No global termination test is required.

For the sequential simulation, store the distinct colors of the already
colored neighbors of $v$ in a hash table.  Repeatedly draw a uniform color
from $[\numclr]$ until a color outside the table is obtained.  This returns a
uniform member of $\cP_t(v)$.  If the table contains $r\leq\deg(v)$ colors,
the expected number of draws is
\[
\frac{\numclr}{\numclr-r}
=1+\frac{r}{\numclr-r}
\leq r+1
\leq\deg(v)+1.
\]
Constructing the table and processing all conflicts and announcements takes
expected $O(\deg(v)+1)$ time at $v$ per iteration.  Summing over all vertices
gives expected $O(m+n)$ time per iteration.  The theorem follows.
\end{proof}

\begin{remark}[Balance in expectation]
\label{rem:balance-expectation}
Algorithm~\ref{alg:slack-greedy} treats the palette symmetrically: available
palettes are complements of neighbor colors, proposals are uniform, and
acceptance is a pure equality test.  Hence relabeling the colors by any
permutation $\pi$ of $[\numclr]$ maps executions to equally likely
executions whose colorings are related by $\pi$, so $\pi\circ\clr$ and
$\clr$ have the same law.  Consequently, all palette colors have the same
expected frequency,
\[
\Exp[\freq(i)]=\frac{n-\Exp[|U_T|]}{\numclr}
\;\geq\;\frac{n-n^{-(\lambda+2)}}{\numclr}
\qquad\text{for every }i\in[\numclr],
\]
where $U_T$ is the uncolored set at termination (the estimate uses
$\Exp[|U_T|]\leq n\cdot\theta^T\leq n^{-(\lambda+2)}$).  The one-sidedness of
Theorem~\ref{thm:slack-greedy} thus reflects the absence of a concentration
argument below the mean, not any systematic imbalance of the process itself.
\end{remark}

We next express the general result in terms of a fractional palette slack.

\begin{corollary}
%[Fractional palette slack]
\label{cor:slack-greedy-fractional}
Fix $0<\rho\leq1$ and put
$\numclr=\Delta+\lceil\rho(\Delta+1)\rceil$.  With probability at least
$1-n^{-\lambda}$, Algorithm~\ref{alg:slack-greedy} computes a proper
upper-balanced coloring with
$\numclr\leq(1+\rho)(\Delta+1)$ colors and frequency cap
\[
O\!\left(
(\lambda+1)\left(\frac{\sigma\lg n}{\rho}+\lg n\right)
\right).
\]
It runs in $O((\lambda+1)\lg n)$ rounds in the \CONGEST\ model and in expected $O((m+n)(\lambda+1)\lg n)$ sequential time.
\end{corollary}

\begin{proof}
The palette slack satisfies
$s\geq\rho(\Delta+1)$.  Equations~\eqref{eq:slack-T-asymptotic}
and~\eqref{eq:slack-cap-parameters} give
\[
\mu_{p,\numclr}
\leq\frac{nT}{\rho(\Delta+1)}
=O\!\left(\frac{(\lambda+1)\sigma\lg n}{\rho}\right),
\qquad
R_{\mathrm{prop}}=O((\lambda+1)\lg n).
\]
Using
$\sqrt{2\mu_{p,\numclr}R_{\mathrm{prop}}}
\leq\mu_{p,\numclr}+R_{\mathrm{prop}}/2$
in Eq.~\eqref{eq:slack-frequency-cap} proves the claim.
\end{proof}

As in Eq.~\eqref{eq:slack-frequency-cap}, the cap of
Corollary~\ref{cor:slack-greedy-fractional} should be read together with the
deterministic bound $\freq(i)\leq n$, which dominates unless $\Delta$
exceeds a constant multiple of $(\lambda+1)\lg n/\rho$.

The following endpoint lets the fractional slack vanish with $n$.

\begin{corollary}
%[Near-exact palette]
\label{cor:slack-greedy-near-exact}
Let $n\geq3$ and $L=\lceil\ln n\rceil$.  There is an absolute constant
$C_0>0$ such that, with probability at least $1-n^{-\lambda}$, Algorithm
~\ref{alg:slack-greedy} computes a proper coloring with
\[
\numclr=\Delta+\left\lceil\frac{\Delta+1}{L}\right\rceil
\leq\left(1+\frac1L\right)(\Delta+1)
\]
colors and
\[
\freq(i)\leq
\min\!\left\{
n,\,
C_0(\lambda+1)(\sigma\lg^2 n+\lg n)
\right\}
\qquad\text{for every used color }i.
\]
Its expected sequential running time is $O((m+n)(\lambda+1)\lg n)$, and its running time in the \CONGEST\ model is $O((\lambda+1)\lg n)$.
\end{corollary}

\begin{proof}
Here $s=\numclr-\Delta\geq(\Delta+1)/L$.  Hence
\[
\mu_{p,\numclr}
=\frac{pnT}{s}
\leq\frac{nTL}{\Delta+1}
=\sigma TL
=O((\lambda+1)\sigma\lg^2 n).
\]
Also $R_{\mathrm{prop}}=O((\lambda+1)\lg n)$.  Substitution in
Eq.~\eqref{eq:slack-frequency-cap}, together with
$\sqrt{2\mu R}\leq\mu+R/2$, gives the stated absolute constant $C_0$.
The palette and running-time claims follow from the definitions and
Theorem~\ref{thm:slack-greedy}.
\end{proof}

For fixed $\lambda$, Corollary~\ref{cor:slack-greedy-near-exact} gives an
$o(n)$ cap whenever $\Delta/\lg^2 n$ tends to infinity.  For smaller
$\Delta$, the exact minimum with $n$ remains essential and prevents the
asymptotic expression from being misread as an improvement over the trivial
cap.

For comparison, increasing the palette further gives a sublogarithmic-round
endpoint.

\begin{corollary}
%[High-slack endpoint]
\label{cor:slack-greedy-sublog}
For fixed $\lambda\geq1$ and all sufficiently large $n$, one setting of
Algorithm~\ref{alg:slack-greedy} uses
$O((\Delta+1)\lg n/\lg\lg n)$ colors, has frequency cap
$O(\sigma+\lg n)$, and runs in $O(\lg n/\lg\lg n)$ rounds in the \CONGEST\ model and in expected $O((m+n)\lg n/\lg\lg n)$ sequential time, with probability at least $1-n^{-\lambda}$.
\end{corollary}

\begin{proof}
Set
\[
\numclr=\Delta+
\left\lceil\frac{(\Delta+1)\ln n}{\ln\ln n}\right\rceil.
\]
For all sufficiently large $n$, we have $\numclr\geq2\Delta$, so
$p=1$, $\theta=\Delta/\numclr$, and
$T=O(\lg n/\lg\lg n)$.  Moreover,
\[
\mu_{p,\numclr}
=\frac{nT}{\numclr-\Delta}
=O(\sigma),
\qquad
R_{\mathrm{prop}}=O(\lg n).
\]
The claim follows from Theorem~\ref{thm:slack-greedy} and
Eq.~\eqref{eq:slack-frequency-cap}.
\end{proof}

For the palette of Corollary~\ref{cor:slack-greedy-sublog}, the ideal
frequency is $n/\numclr=\Theta(\sigma\lg\lg n/\lg n)$; the cap
$O(\sigma+\lg n)$, stated against $\sigma=n/(\Delta+1)$ for comparability
with Corollaries~\ref{cor:slack-greedy-fractional}
and~\ref{cor:slack-greedy-near-exact}, therefore exceeds the ideal frequency
of its own palette by a factor $\Theta(\lg n/\lg\lg n)$.

%%%%%%%%%%%%
\section{Discussion}
\label{sec:discussion}

The two procedures avoid global counting by paying different prices.
\SamplerBalance\ retains a genuine two-sided guarantee but inherits the
roughly $2\Delta$ palette threshold of its sampling black box.
\SlackGreedy\ approaches $\Delta+1$ colors but controls only the upper tail
of the color-class sizes.  
%In particular, it may create singleton classes and is therefore not near-equitable in the standard two-sided sense.
This one-sidedness is a limitation of our proof rather than an established property of the algorithm: by Remark~\ref{rem:balance-expectation}, every color class has the same
expectation, so what is missing is a concentration argument below the mean,
not balance in expectation.  Whether \SlackGreedy, or a mild variant of it,
is in fact near-equitable in the two-sided sense --- say, in the load regime
$\sigma_{\numclr}=\Omega((\lambda+1)\lg n)$ in which two-sided guarantees
are informative at all --- is an open question.

The near-exact palette of Corollary~\ref{cor:slack-greedy-near-exact} is most
useful in the high-degree regime.  For low-degree graphs, the deterministic
minimum with $n$ is the meaningful part of the bound.  Obtaining a
diameter-independent two-sided guarantee with a palette close to
$\Delta+1$, or obtaining a nontrivial upper cap for all degree regimes,
remains open.
Another open question, discussed in Remark~\ref{rem:budget}, is
whether the $\Theta(\lg n)$ gap between the deterministic proposal budget
$\mu_{p,\numclr}$ and the expected proposal load $O(n/(\numclr-\Delta))$ can
be closed, which would sharpen the caps of
Corollaries~\ref{cor:slack-greedy-fractional}
and~\ref{cor:slack-greedy-near-exact} by a $\lg n$ factor.

Another natural direction is to replace uniform proper-coloring sampling by a
more direct process whose stationary distribution penalizes overloaded
colors.  Such a process might trade some of the current palette slack for
stronger balance, but its local implementation and mixing time would require
additional ideas.

%%%%%%%%%%%%%%%%%%%%%%%%%%%%%%%%%

\end{document}